%% file: SubmissionSingleColumn.tex
\documentclass[conference,a4,10pt]{IEEEtran}
\IEEEoverridecommandlockouts
\usepackage{amsmath}
\usepackage{amsthm}
\usepackage{amssymb}
\usepackage{verbatim}
\usepackage{graphicx}
\usepackage{tikz}
\usepackage{pgfplots}
\usetikzlibrary {positioning}
\usetikzlibrary{shapes,snakes}
\usetikzlibrary{arrows,calc}
\usepackage{relsize}
\usetikzlibrary{patterns}
 \usepackage{algorithm}
 \usepackage{algorithmic}
\newtheorem{lemma}{Lemma}
\newtheorem{rem}{Remark}
\newtheorem{theorem}{Theorem}

\newtheorem{definition}{Definition}

\theoremstyle{definition}

\floatname{algorithm}{Algorithm}

\allowdisplaybreaks[4]

\input{input.tex}

\def\sfE{\mathsf{E}}
\begin{document}





\title{Online Energy Harvesting Problem Over An Arbitrary Directed Acyclic Graph Network}

%
%
%
\author{\IEEEauthorblockN{Rahul Vaze\\}
\IEEEauthorblockA{
Tata Institute of Fundamental Research\\
 Homi Bhabha Road, Mumbai 400005\\
Email: rahul.vaze@gmail.com\\}
\and 
\IEEEauthorblockN{Sibi Raj B Pillai \\ }
\IEEEauthorblockA{Indian Institute of Technology  Bombay\\
 Mumbai, India\\
Email:bsraj@ee.iitb.ac.in }

}

\maketitle

\def\tp{t^{\prime}}
\input{Introduction}

\input{SystemModel}

\input{Layer}

\input{RateMax3}

\def\Ni{\mathcal N_i}
\input{multlayerVer1}

\end{document}

%% file: input.tex
\usepackage{amsfonts}
\usepackage{times}
\usepackage{latexsym}
\usepackage{amssymb}
\usepackage{amsmath}
\usepackage{cite}
\usepackage{verbatim}
\def\bb0{{\mathbb{0}}}

\def\bb{{\mathbf{b}}}

\def\br{{\mathbf{r}}}

\def\b0{{\mathbf{0}}}

\def\b1{{\mathbf{1}}}

\def\cL{\mathcal{L}}

\def\cN{\mathcal{N}}

\def\cP{\mathcal{P}}

\def\sfA{\mathsf{A}}

\def\sfE{\mathsf{E}}

\def\sfO{\mathsf{O}}
\def\sfP{\mathsf{P}}

\def\sfc{{\mathsf{c}}}

\def\sf0{{\mathsf{0}}}


%% file: Introduction.tex
\begin{abstract}
A communication network modelled by a directed acyclic graph (DAG) is considered, over which a source 
wishes to send a specified number of bits to a destination node. Each node of the DAG is powered by a
separate renewable energy source, and the harvested energy is used to facilitate the source destination 
data flow. The challenge here is to find the optimal rate and power allocations across time for each node 
on its outgoing edges so as to minimize the time by which the destination receives a
 specified number of bits. 
An online setting is considered where an algorithm only has causal information about the energy arrivals.
Using the competitive ratio as the performance metric,  i.e. 
the ratio of the cost of the online algorithm and  the optimal offline algorithm, maximized over all inputs, a {\it lazy} online algorithm with a competitive ratio of $2+\delta$ for any $\delta>0$  is proposed. Incidentally,
$2$ is also a lower bound to the competitive ratio of any online algorithm for this problem.  
Our lazy online algorithm is described and analyzed via defining a novel max-flow problem over 
a DAG, where the  rate on the subset of outgoing edges of any node are related/constrained. 
An optimal algorithm to find max-flow  with  these constraints is also provided, which may be
of independent interest.
\end{abstract}
\section{Introduction}
Enabling communication nodes to harvest energy from nature makes them robust, and enhances their lifetime. Moreover, it 
also makes the communication {\it green}. One challenge, however, is that the energy arrivals from nature are 
inherently uncertain, and the communication algorithms have to adapt to the randomness of energy availability. 
This paradigm (called energy harvesting or EH) presents fresh challenges in designing optimal communication 
algorithms, and in the past few years, there has been lot of work towards that direction.

In this paper, we consider a source-destination pair that is connected via an arbitrary directed acyclic graph (DAG).  The DAG models a communication network setting where direct communication is
possible from each node to its first hop neighbor, i.e. via the edge. We assume that the edges 
are orthogonal, i.e. the links do not mutually  interfere. However, the DAG topology
is otherwise arbitrary. 
%
%
%
%
Each node of the DAG is powered by EH, where the amount and the time 
instants of energy arrivals are assumed to be arbitrary. We consider the online setting, where any 
algorithm has only causal information about energy arrivals, and the objective of the algorithm is to transport 
a specified number of bits from the source to the destination in as minimum a time  as possible, using the 
energy arrivals at the respective nodes of the DAG. We call this problem as the {\it delay-minimization} problem.

In prior work, delay-minimization as well as the related rate-maximization problems have been considered for a 
small number of nodes \cite{basic Ulukus,mac, comp_ratio, off_on, Ayfer16, Ashwini, D(onlinemac),onlinetwoway(F), Ulukus2011}, \cite{vazeJSAC, Ozgur2015,Stuber2015,Ulukus2016} such as point-to-point, a single unicast with multiple relays, 
MAC channel with multiple transmit nodes etc. Prior work primarily addresses
the {\it offline} setting, while fewer results are known in the online setting \cite{Ayfer16,Ashwini, D(onlinemac),onlinetwoway(F), vazeJSAC}.

By offline, we mean that the algorithm has non-causal information about all energy arrivals in the future. 
To the best of our knowledge there has been no work on the delay minimization problem for an arbitrary DAG  either in 
the offline or the online setting, as considered in this paper.
The main challenge in a network setting is that the optimal energy utilization at different nodes is inter-dependent, 
making the problem challenging for an arbitrary network topology.  

To quantify the performance of an online algorithm, the concept of competitive ratio is used, 
that is defined as 
the ratio of the cost of an online algorithm and the optimal offline algorithm, maximized over the worst case input. 
This metric might appear too pessimistic, however, in prior work there has been success in finding online 
algorithms that have competitive ratio of at most $2$ for the delay-minimization  in a two 
node problem \cite{comp_ratio, off_on}. Moreover, in \cite{off_on} it is shown that 
no online algorithm can have competitive ratio better than $2$, even for a two node network.

In this paper, we first propose an online algorithm for the delay-minimization problem over an 
arbitrary DAG and show that its  competitive ratio is at most $2+\delta$ for any $\delta>0$.
The computational complexity of the algorithm is $\sfO(\log(1/\delta)), 0< \delta < 1$
(Big O notation). 
Thus, close to optimal competitive ratio can be obtained by choosing a small $\delta$ while 
paying a very minor penalty in the complexity, since the competitive ratio is lower bounded by $2$ even for a two node network \cite{off_on}.  
We would like to point out that even the optimal offline algorithm for the delay 
minimization problem over a DAG with EH nodes is unknown, and challenging to find,
 given the arbitrary network topology. Nevertheless, we show that a suitable lazy algorithm is constant competitive in the online setting, notwithstanding the fact that the optimal offline algorithm is unknown. 

In order to explain the competitive performance of our online lazy algorithm,
assume that the optimal offline algorithm, starting from time $0$, completes the transmissions in $t_0$ (unknown) units of time.  The main idea of the online algorithm is to estimate $t_0$ reasonably closely. Let this estimate be ${\hat t}_0 = t_0 + \delta$. Suppose there exists an online  algorithm which can 
transfer the specified number of bits while transmitting only in the
time duration $[{\hat t}_0,t_0+{\hat t}_0]$ using only the energy that arrives in the interval 
$[0, {\hat t}_0]$. Clearly, such an algorithm is energy feasible, and moreover its competitive ratio is $\frac{t_0 + {\hat t}_0}{t_0}$. This will lead to a $2+\delta$-competitive algorithm. 

In short, the online 
algorithm needs to ensure two things, \textit{viz.} (i) ${\hat t}_0 = t_0 + \delta$ for some fixed $\delta$, (ii) transport $B_o$ bits while transmitting only within the time interval 
$\left[\hat t_0, \hat t_0 + t_0 \right]$.
To accomplish both these tasks, we take recourse to a {\it novel}
max-flow problem over a DAG. Recall that with classical max-flow problem, 
given the capacity for each edge of the network, the maximum flow possible between
a source node and its destination is to be determined. A more generalized version of 
this is when there are constraints on different subsets of outgoing edges from 
each node. For example, in the \emph{polymatroidal} max-flow problem \cite{Lawler, Chekuri2015}, 
the set of rates possible on outgoing links of any node  are defined as the intersection 
of hyperplanes. However, a close inspection of our problem reveals that the rate constraints
are not polymatroidal. In particular,  letting rate to be logarithmic in power using the Shannon rate formula, if the out-degree of a node is $2$ with total power $P$, then the rate constraints on the two outgoing links will result in a 
region $(r_1,r_2) = (\log(1+\alpha P), \log(1+(1-\alpha) P))$ as shown in Fig. \ref{fig:rateregion} that is non-polymatroidal, whose boundary is traced by $0\leq \alpha \leq 1$.
Thus the set of possible rates on two outgoing
links subject to a common power constraint is not polymatroidal. This calls for an alternate
approach to the max-flow problem here.

We show that if the offline optimal algorithm can communicate $B_o$ bits by time $t_0$ 
(clearly using only the energy that has arrived till time $t_0$), the optimal 
max-flow solution can maintain a rate of at least $\frac{B_0}{t_0}$ for the 
time interval $[t_0, 2t_0]$, while using only the 
energy that has arrived till time $t_0$. 
Thus, employing the max-flow solution from time $t_0$ till $2t_0$ can also finish transmission 
of $B_o$ bits. The only remaining task is to estimate $t_0$ closely. Fortunately, one can 
find an upperbound to $t_0$ by solving a max-flow problem at every energy arrival.
The upperbound itself is at most $2 t_0$,  thus a further line search can  find the actual $t_0$.
The latter has complexity  $\sfO(\log(\frac{t_0}{\delta}))$ to obtain an estimate within $t_0+\delta$.



The defined max-flow problem is important in its own right since it advances the literature on flow maximization. When compared to the classical and the polymatroidal max-flow, the novel properties of the considered max-flow problem include that the min-cut capacity is not equal to the max-flow. Thus, usual augmenting path algorithms~\cite{Lawler} are insufficient to find the optimal flow.

Our contributions are as follows:
\begin{itemize}
\item For an arbitrary DAG, we present a $2+\delta$ (for any $\delta>0$)-competitive online algorithm for the delay minimization problem. Since $2$ is a lower bound on the competitive ratio for any online in a $2$-node network case, choosing $\delta$ small gives an almost optimal online algorithm for the DAG network.
\item We define and solve a novel max-flow problem over a DAG network with non-polymatroidal outflow constraints, which
is of independent interest in max-flow literature with edge and node constraints. 
Using the fact that a DAG with orthogonal links is equivalent to a layered network, we propose an iterative algorithm for solving the max-flow problem on the layered network, that tries to find max-flow on each layer recursively, and is shown to be optimal.
\item One limitation of our results on DAGs is that all edges are assumed to be orthogonal. For a special case of a DAG, a layered network, we show that the max-flow problem can be solved even while incorporating
 polymatroidal interference constraints on incoming edges at nodes, e.g. non-orthogonal MAC constraints.
\end{itemize} 

Rest of the paper is organized as follows.
After detailing the system model in Section~\ref{sec:model}, 
we first connect the delay-minimization problem over a DAG having EH nodes to a
network max-flow problem in Section \ref{sec:nonpoly}. We then
show how the optimal solution of the max-flow problem can be used to define a lazy online 
algorithm for the original delay-minimization problem. 
In addition, we also analyze the competitive ratio of the proposed algorithm.
Thereafter, in Section~\ref{sec:max:flow}, we derive an optimal algorithm 
to solve the non-polymatroidal max-flow problem, a challenging task in its own right, 
since unlike the classical/polymatroidal case, here the max-flow may not be the same as min-cut.
We first consider the
specific case of a three hop layered network in Section~\ref{sec:two:layer}, and generalize  to arbitrary number of layers in Section~\ref{sec:mult:layer}. We present some numerical results to illustrate the performance of the proposed algorithm in Section \ref{sec:sim}. Finally, Section~\ref{sec:conc}
concludes the paper.

%% file: SystemModel.tex
\def\lzo{\mathbf{\textsc{LazyOnline}}}

\section{System Model and Objectives} \label{sec:model}

Consider a directed communication network described by a graph $G=(V,E)$, where $V$ is the set of nodes, 
and $E$ the directed set of edges, each connecting a pair of nodes. 
The graph $G$ is assumed to be acyclic, thus, $G$ is a directed acyclic graph (DAG). 
Each node is assumed to be full-duplex. The half-duplex case is fundamentally different and more challenging even for a two-node network \cite{Ashwini}. 
For node~$k$, let $I_k$ denote the set of nodes from which there are edges incident to it, and $O_k$ represent the
set of nodes to which there are outgoing edges from $k$. Direct communication is possible between any pair of nodes 
only if they have an edge between them. Moreover, communication over distinct edges is 
orthogonal and does not interfere with each other. One example of a considered network is provided in Fig. \ref{fig:eg:model}, where node $1$ (the source) wishes to communicate with node $6$ (the destination) via nodes $2,3,4,5$ that are connected via directed orthogonal edges.
We also consider some generalization to
non-orthogonal links in Section \ref{sec:lay},  for a special topology of the DAG
called the layered network.

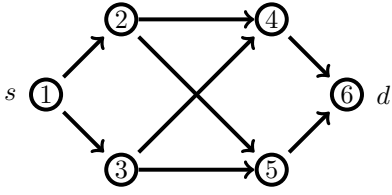
\begin{figure}[h]
\begin{center}
\begin{tikzpicture}[line width=1.5pt, every line/.style={->}, scale=1,
lab/.style={yshift=0.25cm}]
\coordinate (c1) at (0,0); \coordinate (c2) at (1,1); \coordinate (c3) at (1,-1);
\coordinate (c4) at (3,1); \coordinate (c5) at (3,-1); \coordinate (c6) at (4,0);
\foreach \i in {1,...,6}{
\draw[fill=white] (c\i) circle(0.2cm) node (n\i) {$\i$};
};
\draw (n1) edge[->] (n2) (n1) edge[->] (n3) (n2) edge[->] (n4) (n2) edge[->] (n5) 
	(n3)edge[->](n4) (n3)edge[->](n5) (n4) edge[->] (n6) (n5)edge[->](n6);

\draw (n1) node[above,lab]{} (n2) node[above,lab]{} (n3) node[above,lab]{}
 	(n4) node[above,lab]{} (n5) node[above,lab]{} 
	(n6) node[right, xshift=0.25cm]{$d$} (n1) node[left, xshift=-0.25cm]{$s$};
\end{tikzpicture}
\end{center}
\caption{Example network \label{fig:eg:model}}
\end{figure}

Each node in the network harvests energy from nature to power its communication and stores that 
in a battery of size $C$. Following majority of prior work, for analytical simplicity, we assume 
that $C$ is large enough such that no charge overflow happens~\footnote{The typical battery size for practical EHNs ranges between $200$ mAh-
$2500$ mAh \cite{Batteryjust}. A $200$ mAh capacity battery can deliver $720$ J of energy at
a nominal voltage of 1 V. Also, using a small solar panel, at $66 \%$ efficiency,
NiMH batteries receive $1.3$ mJ of energy per $100$ ms slot. Thus, with two
hours of sunlight, the typical battery size, normalized with respect to Es,
equals $5.33\times 10^5$. Hence, in practice, a battery size of $1000$ is very small.}. 
We assume that an amount $E_{kj}, j \geq 1$ Joules of energy arrives at node~$k$ on the time 
instant  $\tau_{kj}$. Let $\sfA_k(t) = \sum_{\tau_{kj} \le t} E_{kj}$ be the total 
accumulated energy by node $k$ till time $t$. The quantity, as well as time instant, 
of energy arrival is assumed to be  arbitrary, and can even be chosen by an adversary. 
W.l.o.g we will let $\tau_{kj}$ to be  increasing in $j$. The information on all the energy 
arrival processes at different nodes is assumed to be causally available at a central location, 
since we are interested in a centralized solution.


We assume that the rate of transmission over a directed edge $e=(u,v)$ of $G$  when 
node $u$ transmits with power $P$ towards node $v$ is concave in $P$. 
In particular, for ease of exposition, we use  
a logarithmic rate given by the Shannon formula for a normalized AWGN channel,
\begin{align} \label{eq:shan:rate}
r(P) = \log (1 + P) \text{ bits/sec/Hz}.
\end{align} 
All the results presented in the paper hold as long as the  rate function is concave in $P$. 

Let node~$k$ transmit power $P_{kl}(t)$ on edge $(k,l), l\in O_k$ towards node~$l$ at 
time $t$.  The total energy $e_k(s)$ expended by node $k$ till time $s$ is then given by
$$
e_k(s)  = \int_0^s \sum_{l \in O_k} P_{kl}(t) dt.
$$ 
Similarly, let $\sfA_k(s)$ denote the total energy arrived at node~$k$ till instant~$s$. Energy causality
constraints will imply that
$$
e_k(t) \leq \sfA_k(t) , \forall t. 
$$
By \eqref{eq:shan:rate}, the instantaneous rate on the edge $(k,l)$ becomes~\footnote{The validity of this formula in a practical setting is justified while having a sufficient bandwidth
for communication, 
this makes coding and decoding possible within the time-scales of interest~\cite{Telatar95}.} 
$$
r_{kl}(t) = \log (1 + P_{kl}(t)), \forall t.
$$
Let $B_{kl}(s)$
denote the total amount of data (in bits) transported by node~$k$ to its neighboring node~$l$ in the time interval $(0,s]$.
$$
B_{kl}(s) = \int_0^s  r_{kl}(t) dt.
$$

%
%

%

We now define the main objective of this paper, to solve the delay-minimization problem, defined as follows. Consider a source-destination pair $(s, d)$ belonging to $G$. 
Source $s$ wishes to send $B_o$ bits to the destination node $d$ over the edges of graph $G$, and the problem is to minimize 
the time by which $B_o$ bits are received by node~$d$. 

%
%
%
 The delay minimization problem can be written as the following
optimization problem with respect to the power allocation function $P_{kv}(t), v \in O_k$ for each node $k\in V$.
\vspace{0.1in}
\hrule
\begin{align}
\label{eq:opt:1}& \ \ \ \ \ \ \ \ \ \ \min  \ \ T \\
\label{eq:optcond:1} 
\text{s.t.} &\sum_{l \in I_k} B_{lk}(t) \geq \sum_{l \in O_k} B_{kl}(t), \forall \ t, \ \forall \ k \in V \backslash \{s,d\} \\ \label{eq:optcond:2}
				&\sum_{l \in O_{s}} B_{sl}(T) = \sum_{l \in I_{d}} B_{ld}(T) = B_o, \\\label{eq:optcond:3}
 & e_{k}(t) \leq \sfA_k(t), P_{kl}(t)\ge 0, k \in V, l \in O_k.
\end{align}
\hrule
\vspace{0.1in}
\vspace*{0.25cm}
\noindent Here \eqref{eq:optcond:1} denotes the flow conservation constraint for each node other than the source and the destination, i.e., the outgoing flow is at most the incoming flow, \eqref{eq:optcond:2} captures the out-flow and in-flow condition 
for the source and destination since $B_o$ bits are needed to be transported, and \eqref{eq:optcond:3} captures 
the energy neutrality constraint for each node. Notice that the above problem formulation is common
to both offline and online schemes. The former can optimize using the transmission schedules
using the non-causal knowledge of all energy arrival
processes, whereas the latter has to make decisions based on the causal knowledge of the arrival process.

All logarithms in this paper are with respect to base $2$. We will denote $|U|$ for the cardinality of the set $U$.

\begin{table}
\caption{Notation Table}
\begin{tabular}{|c|p{7cm}|}
\hline
Symbol & Notation \\
\hline 
$I_k$ & Set of nodes in layer $\cL_{l-1}$ that have an edge to node $k$ in layer $\cL_{l}$ \\ \hline  
$O_k$ & Set of nodes in layer $\cL_{l+1}$ that have an edge to node $k$ in layer $\cL_{l}$ \\ \hline  
$r_i$ & For a node $i$ of layer $k$, the sum-rate out of node $i$ towards nodes of layer $k+1$  \\ \hline 
$R_k$ & For a layer $k$, the sum of sum-rate $r_i$ out of all all nodes $i$ in layer $k$ \\ \hline 
$f_i$ & For a node $i$ of layer $k$, the total incoming rate from all nodes of layer $k-1$ \\ \hline  
$g_j$ & For a node $i$ of layer $k$, the sum-rate going out of node $i$ towards nodes in layer $k+1$  \\ \hline  
$U$ & set of nodes $i$ of layer $k$ for which $r_i < f_i$  \\ \hline  
$N(S)$ & For a set of nodes $S$ of layer $k$, $N(S)$ is the set of nodes of layer $k-1$ that have an edge to some node in $S$\\ \hline  \hline
\end{tabular}
\end{table}

\section{Online Algorithm and Competitive Ratio}\label{sec:nonpoly}
Solving Problem \eqref{eq:opt:1} is complicated even in the \textit{offline} setting, where all the energy arrivals are known in advance. 
In this paper, we consider the online setting, i.e., any algorithm can only use causal information about the energy arrivals and wants to solve \eqref{eq:opt:1}. To describe the online setting, we need the following notation.

Recall that $\tau_{k_j}$ denote the energy arrival instants at node~$k$. Let us create
a lexicographically increasing sequence of tuples 
$$
\sigma = \{(\tau_{k_j}, k, E_{k_j}), \forall k \in G\}.
$$
Thus, $\sigma$ represents the combined energy arrival sequence on all nodes in the network. 
Let $T_{{\cal A}}(\sigma)$ and $T_{\text{off}}(\sigma)$ be the respective 
total transmission completion times solving \eqref{eq:opt:1}, for the online algorithm ${\cal A}$ 
and the optimal offline algorithm (which will remain unknown), respectively. 
We use the competitive ratio as the performance metric for online algorithms, 
that is defined for an online algorithm ${\cal A}$ as 
\begin{equation}\label{def:cr}
\mu_{\cal A} = \max_\sigma \mu_{\cal A}(\sigma) = \max_\sigma\frac{T_{{\cal A}}(\sigma)}{T_{\text{off}}(\sigma)},
\end{equation}
where the maximum is over all possible energy arrival sequences $\sigma$, that can be chosen even adversarily. The optimal competitive ratio $\mu^\star$ is defined as $\mu^\star = \min_{\cal A}\mu_{\cal A}$ and an algorithm ${\cal A}^\star$ is called an optimal online algorithm  ${\cal A}^\star=\arg \min_{\cal A}\mu_{\cal A}$, i.e, if it achieves the optimal competitive ratio.
Our objective is to find an optimal online algorithm which achieves the minimum competitive ratio, since by definition, an online algorithm with low competitive ratio has good
performance even against adversarial inputs. 

Towards this direction, we first define a related rate maximization problem, 
which turns out  useful while proposing an online algorithm for solving \eqref{eq:opt:1}. 
For a given time $\tp$,  let us construct a scheme in which node~$k$ only uses the accumulated 
energy $\sfA_k(\tp)$. Moreover, node $k$ is also restricted to not transmit 
at all till time $\tp$, and transmit with equal power over time $[\tp, 2\tp]$ using the energy 
$\sfA_k(\tp)$. Thus $P_{kl}(t) = P_{kl}$ for 
the time interval $[\tp, 2\tp]$, and $\sum_{l \in O_k} P_{kl}  = \frac{\sfA_k(\tp)}{\tp}:= \sfP_k(t')$.
The energy neutrality constraint is clearly met at node~$k$, however the power allocation  
$P_{kl}, l \in O_k$ can be further optimized. We can now choose the powers $P_{kl}, \ell \in O_k, \forall k \in G$  to maximize the source destination flow in the interval $[\tp,2\tp]$. 
Thus for each $t=\tp$ we can define a (max-flow) rate maximization problem as follows:
%
\hrule
\begin{align}
 \label{eq:opt:2}
\max 	&  \ \ \ \ \ R(t') \\
\label{eq:optratecond:3}
\text{s.t.}\  & r_{kl} = \log(1+P_{kl}), \sum_{l \in O_{k}} P_{kl} \leq \sfP_{k}(t'), \\
\label{eq:optratecond:1} 
&\sum_{l \in I_k} r_{lk} \geq \sum_{l \in O_k} r_{kl}, \ \forall \ k \in V \backslash \{s,d\} \\ \label{eq:optratecond:2}
				&\sum_{l \in O_{s}} r_{sl} = \sum_{l \in I_{d}} r_{ld} = R(t'), P_{k\ell} \geq 0, \ell \in O_k,
\end{align}
\hrule
\vspace{0.1in}
\vspace*{0.25cm}
\noindent where in \eqref{eq:optratecond:3} $r_{kl}$ is the rate achieved on each of the outgoing links $l \in O_k$, while 
\eqref{eq:optratecond:1} and \eqref{eq:optratecond:2} capture the flow conservation constraints at each node. 

Recall that a {\it max-flow} problem over a given directed graph with specified edge capacities is to find the largest rate of commodity that can be transported from a given node (source) to another (destination) that respects the edge capacities \cite{AlgorithmsBook}. 
Essentially, Problem \eqref{eq:opt:2} is a single source-destination max-flow problem  that maximizes the instantaneous rate (at time $t'$) 
from source to destination if the power used by node~$k$ is fixed to be $\sfP_{k}(t')$ for 
$k\in G$. The optimal power allocation by each node on its outgoing links subject to
a sum power constraint of $\sfP_{k}(t')$ is to be found.
%
When compared to Problem \eqref{eq:opt:1}, Problem \eqref{eq:opt:2} does not have a time based decision component, 
since transmit power allocation of node $k$ is fixed for the whole duration $[\tp, 2\tp]$, 
and is thus simpler to solve.

\begin{rem}For the rest of the paper, we proceed as follows. As noted before, the delay minimization problem problem involves finding optimal power transmission strategies for each node that are a function of time, which is challenging. 
To simplify the problem, we have defined an intermediate max-flow problem \eqref{eq:opt:2} that uses a fixed power transmission strategy, and we show that if we can solve the max-flow problem optimally, then we can derive online algorithm for the delay minimization problem with competitive ratio of $2+\delta$ as shown in Lemma \ref{lem:compratio}. 
The solution to the max-flow problem for a two-layer network is provided in Section \ref{sec:two:layer}, which is then extended for arbitrary number of layers in Section \ref{sec:mult:layer}.
\end{rem}
Suppose, for any $\tp$, we can solve  \eqref{eq:opt:2} to find  the optimal rate as $R^\star(\tp)$. 
Lemma~\ref{lem:connection} connects Problems \eqref{eq:opt:1} and \eqref{eq:opt:2}, in turn enabling
a \emph{lazy} online algorithm to solve \eqref{eq:opt:1}.

\begin{lemma}\label{lem:connection} For a given energy arrival sequence $\sigma$, let  
$T_{\text{off}}(\sigma)$ be the optimal time obtained by solving \eqref{eq:opt:1}. 
Then $T_{\text{off}}(\sigma)R^\star(T_{\text{off}}(\sigma)) \ge B_o$.
\end{lemma}
\begin{proof} 
Notice that since the optimal offline scheme only employed energy collected till $T_{\text{off}}(\sigma)$,
we can as well restrict node $k$ of the network to use only the energy $\sfE_k(T_{\text{off}}(\sigma))$ that was harvested till time $T_{\text{off}}(\sigma)$. 
Out of $\sfE_k(T_{\text{off}}(\sigma))$, if node $k$ spends an energy $\sfE_{k\ell}(T_{\text{off}}(\sigma))$ on its outgoing link $\ell \in O_k$, then since $\log$ is a concave function, the number of bits sent by node $k$ on its outgoing link $\ell \in O_k$ with the optimal offline algorithm is at most $$B_{ub}(k,\ell) = T_{\text{off}}(\sigma) \log \left(1+\frac{\sfE_{k\ell}(T_{\text{off}}(\sigma))}{T_{\text{off}}(\sigma)}\right)$$ such that $$\sum_{\ell \in O_k} \sfE_{k\ell}(T_{\text{off}}(\sigma)) \le \sfE_{k}(T_{\text{off}}(\sigma)).$$ Thus, the number of bits sent by node $k$ on all its outgoing links is at most 
$$B_{ub}(k) = \max_{E_{k\ell}, \ell \in O_k} \sum_{\ell \in O_k}T_{\text{off}}(\sigma) \log \left(1+\frac{\sfE_{k\ell}(T_{\text{off}})(\sigma)}{T_{\text{off}}(\sigma)}\right),$$
$ \ \forall \ \ell\in O_k$ subject to $\sum_{\ell \in O_k} E_{k\ell}(T_{\text{off}})(\sigma) \le \sfE_k(T_{\text{off}})(\sigma)$.
Thus, summing the bits coming into the destination $\sum_{k\in I_d}B_{ub}(k)$, we get that $B_o \le \sum_{k\in I_d}B_{ub}(k)$.

For node $k$, with reference to Problem \eqref{eq:opt:2}, defining the variable power partition as $P_{k\ell} = \frac{\sfE_{k\ell}(T_{\text{off}})(\sigma)}{T_{\text{off}}(\sigma)}, \ell \in O_k$ and the total power constraint $\sfP_k =\frac{\sfE_k(T_{\text{off}})(\sigma)}{T_{\text{off}}(\sigma)}$, we see that $$B_{ub}(k) = T_{\text{off}}(\sigma) \sum_{\ell \in O_k} r_{k\ell}^\star,$$ where $r_{k\ell}^\star$ is the optimal rate for Problem \eqref{eq:opt:2} with $t' = T_{\text{off}}$.
Since this is true for all nodes of $V$, summing over all nodes in the network, in particular the ones that have directed edges to the destination, $\sum_{\ell \in I_{d}} r^\star_{\ell d} = R^\star(T_{\text{off}}(\sigma))$ that contribute the flow towards the destination, we get 
$\sum_{k\in I_d}B_{ub}(k)= T_{\text{off}}(\sigma) \sum_{k \in I_d} r_{\ell d}^\star$. Thus, we get that 
$T_{\text{off}}(\sigma) R^\star(T_{\text{off}}(\sigma))$ has to be at least as much as $B_o$.
\end{proof}

Thus, if we knew $T_{\text{off}}(\sigma)$ and the solution of Problem \eqref{eq:opt:2} for $t'=T_{\text{off}}(\sigma)$, we could directly use Lemma \ref{lem:connection} to find a feasible solution for Problem \eqref{eq:opt:1}. However, since we do not know the optimal offline algorithm or $T_{\text{off}}(\sigma)$, 
we now define an algorithm (Algorithm~$\lzo$) for finding a suitable time $t'$ (estimate of $T_{\text{off}}(\sigma)$) such that solving for $R^\star(t')$ will result in a feasible solution for Problem \eqref{eq:opt:1}. Notice that $R^{\star}(t')$ depends on $P_k(t'), \forall k$. Let us extend this definition and denote $R^{\star}(t',\Delta t', s)$, as the solution to \eqref{eq:opt:2} 
where the energy available at node $k$ is only the energy accumulated till time $s$, i.e., $\sfA_k(s)$, 
and the algorithm transmits for time interval $[t', t'+\Delta t']$ with equal power $\frac{\sfA_k(s)}{\Delta t'}$ from each node, and the optimisation is over the power allocation  $P_{kl}$ on outgoing edges from $k$ to nodes $l \in O_k$.
Thus each node can only use the energy harvested till time $s$ in evaluating $R^{\star}(t',\Delta t', s)$.  The case of interest here is $s\leq t'$. 
 The definition of $R^{\star}(t',\Delta t', s)$ allows the implementation of a look ahead scheme to find a suitable upper bound to $T_{off}(\sigma)$ in Algorithm~$\lzo$.

\begin{algorithm}
\caption{~$\lzo$~\label{algo:lazy}}

\begin{algorithmic}

\vspace*{0.15cm}
\STATE On the first energy arrival (anywhere in the network) instant, $\min \tau_{kl}$, initialize the time counter $c =\min \tau_{kl}$. \\

\IF{ $\bigl( 2c R^\star(c,2c,c) \ge B_o\bigr)$} 
\STATE Find $t_{min} = \min \{ t' : t' R^{\star}(t') \geq B_o, c \leq t' \leq 2c  \}$.

\STATE Obtain the static power allocation 
$P^\star_{k\ell}, \ell \in O_k, k \in V$ achieving $R^{\star}(t_{min})$ in \eqref{eq:opt:2}. 

\STATE Employ this static power allocation for time $[t_{min}, 2t_{min}]$ and
 send $B_o$ bits. \\

\ELSE
\STATE Update time counter to $c= \min \bigl( 2c, \min\{\tau_{kl}: c < \tau_{kl} \}\bigr)$ (double the counter or go to the next energy arrival instant), 
\STATE
Wait till time $c$, then go to Step I \\

\ENDIF
 
\end{algorithmic}

\end{algorithm}

Notice that the algorithm first checks at time $c$, 
whether the energy that has arrived in interval $[0,c]$ is sufficient to send $B_o$ 
bits within the interval $[c,3c]$ or not, i.e., whether $2c R^{\star}(c, 2c,c) \ge B_o$ or $2c R^{\star}(c, 2c,c) < B_o$. If $2c R^{\star}(c, 2c,c) < B_o$, then it is not possible to transmit $B_o$ bits using energy that has arrived till time $c$ in interval $[c,3c]$ of width $2c$. This also means that if no new energy arrives in interval $[c,2c]$, then it is not possible for any online algorithm to transmit $B_o$ bits in interval $[0,2c]$. Thus, if the next energy arrival $\tau_{kl} > 2c$, 
then $c$ is updated to $2c$.
In case new energy arrives before time $2c$, $c$ is updated to next energy arrival instance. Since the algorithm is online, this means that after checking at $c$, we wait till time $2c$ or the next arrival instance whichever happens earlier and check again whether $2c R^{\star}(c, 2c,c) < B_o$ or not.

On the other hand, if 
$2c R^{\star}(c, 2c,c) > B_o$, then we can employ a line search to
find the parameter $t_{min} \in [c,2c]$ such that  $t_{min} = \min \{ t' : t' R^{\star}(t') \geq B_o, c \leq t' \leq 2c  \}$, and the line search complexity is  
 $\sfO(\log 1/ \delta)$ steps for an accuracy of $\delta \in (0,1)$. 
Note that we are assuming with Algorithm~$\lzo$, that each node using 
the full-duplex mode is able to forward the bits it is receiving on its outgoing links without any delay. We can even account for the forwarding (decoding/encoding) delay within the overall completion time, since the forwarding delay is typically small compared to the bits transmission time.

\begin{lemma}\label{lem:compratio}The competitive ratio $\mu_{\mathsf{Lazy}}(\sigma)$ of online algorithm $\lzo$ is at most $2+\frac{2\delta}{T_{\text{off}}(\sigma)}$, where $\delta$ is the step-size.
\end{lemma}
\begin{proof}
Let the counter be at $c=c^\star$, when the algorithm $\lzo$ breaks.
From Lemma \ref{lem:connection}, $c^\star \le T_{\text{off}}(\sigma) + \delta$.  Thus, using the static power allocation $P^\star_{k\ell}, \ell \in O_k$ found by the solution of \eqref{eq:opt:2} for $R^\star(c^\star)$, $B_o$ bits can be sent to the destination from time $[c^\star, 2c^\star  ]$. Thus the time taken to finish transmission of $B_o$ bits by $\lzo$ is $T_{\text{Lazy}}(\sigma)= 2c^\star$.
Therefore, the competitive ratio of $\mathsf{Lazy}$ is at most 
\begin{align*}\mu_{\mathsf{Lazy}}(\sigma) &= \frac{2c^\star }{T_{\text{off}}(\sigma)} \le \frac{2T_{\text{off}}(\sigma) + 2\delta }{T_{\text{off}}(\sigma)} \le 2+ \frac{2\delta}{T_{\text{off}}(\sigma)}.\end{align*}\end{proof}

\begin{theorem}\label{thm:compratio} Algorithm $\lzo$ is an optimal online algorithm. \end{theorem}
\begin{proof}From \cite{off_on}, it follows that for $G$ where there are only two nodes, the competitive ratio for solving \eqref{eq:opt:1} is lower bounded by $2$. From Lemma \ref{lem:compratio}, choosing the scanning width $\delta$ small enough (controlled by the complexity budget) in algorithm $\lzo$, we can make its competitive ratio arbitrarily close to $2$, completing the proof.
\end{proof}

Note that the choice of $\delta$ will depend on the complexity budget, since the complexity of algorithm $\lzo$ is $O(\log (1/\delta))$.

For the rest of the paper, we concentrate on solving \eqref{eq:opt:2} which is a max-flow problem with logarithmic output utility/flow under sum-power constraint at each node. Since Problem \eqref{eq:opt:2} is of interest on its own, we present a self contained presentation of its solution.

%% file: Layer.tex
\section{Max-Flow Problem} \label{sec:max:flow}
For the  ease of exposition, rather than working with a DAG, we instead consider an equivalent
layered network, where there are $K$ intermediate layers 
between the source and the destination. 
The set of nodes in the intermediate layer $k\in\{1, \dots, K\}$ is denoted as $\cL_k$. 
Source $s$ is at layer~$0$, and the destination is at layer $K+1$.
Notice that in a layered network edges exist only between nodes of adjacent layers, 
i.e. between nodes of $\cL_k$ and $\cL_{k+1}$.
In Fig. \ref{fig:mapping}, we illustrate via a simple example how an equivalent layered network (on the right) can be constructed from a DAG (on the left), where $a \equiv a^{\prime}$,
$b \equiv b^{\prime}$, $d \equiv d^{\prime}$, and $c^{\prime\prime} \equiv c$.
The extra nodes $b^{\prime\prime}$ and $c^{\prime}$ have infinite power (no power constraint). 
Essentially, the idea is that each if a pair of nodes $u,v$ have an edge between them that either belong to the same layer or are in non-adjacent layers, add extra nodes (dummy nodes) 
corresponding to $u,v$ with infinite power to make the network a layered network. For brevity we omit the precise construction, which is immediately clear from the above description.
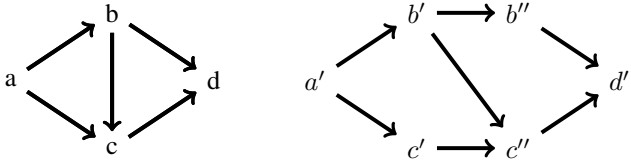
\begin{figure}
\begin{centering}
\begin{tikzpicture}[node distance=2cm,scale=.9, line width=1.5pt]
\label{fig:trans}
\node (A) at (0, 0) {a};
\node (B) at (1.5, 1) {b};
\node (C) at (1.5, -1) {c};
\node (D) at (3, 0) {d};

\draw[->]
  (A) edge (B) (B) edge (C) (C) edge (D) (A) edge (C) (B) edge (D)  ;

\pgftransformxshift{1.0cm}

\node (A') at (3.5, 0) {$a^{\prime}$};
\node (B') at (5, 1) {$b^{\prime}$};
\node (B'') at (6.5, 1) {$b^{\prime\prime}$};
\node (C') at (5, -1) {$c^{\prime}$};
\node (C'') at (6.5, -1) {$c^{\prime\prime}$};
\node (D') at (8, 0) {$d^{\prime}$};

\draw[->]
  (A) edge (B) (B) edge (C) (C) edge (D) (A) edge (C) (B) edge (D)  ;

\draw[->]
  (A') edge (B') (C'') edge (D') (A') edge (C') (B'') edge (D')
  (B') edge (B'') (B') edge (C'') (C') edge (C'');

\end{tikzpicture}
\end{centering}

\caption{Mapping a DAG to a layered network.~\label{fig:mapping}}
\end{figure}

Problem \eqref{eq:opt:2} is essentially a max-flow problem, where, the rate achievable on any subset of outgoing links from any node are constrained, unlike the classical problem where each edge has individual rate constraint/capacity. Even with  
these additional constraints, 
if the outgoing rate constraints are polymatroidal, i.e. defined by intersection of hyperplanes, 
one could use the result from \cite{Lawler} to find the solution. However, the rate constraints of the type 
considered in Problem \eqref{eq:opt:2} are not polymatroidal ones. 
For example, if the out-degree of a node is $2$ with total power $P$, then the rate constraints \eqref{eq:opt:2} will result in a 
region $(r_1,r_2) = (\log(1+\alpha P), \log(1+(1-\alpha) P))$ as shown in Fig. \ref{fig:rateregion} that is non-polymatroidal, whose boundary is traced by $0\leq \alpha \leq 1$.
Thus, Problem~\eqref{eq:opt:2} is in fact a novel problem, which is of independent interest in 
the max-flow literature. Moreover, the flow-conservation constraints \eqref{eq:optratecond:1} and \eqref{eq:optratecond:2} are equal to a difference of $\log$ terms which in general need not result in a convex constraint set. In Lemma \ref{lem:concave}, we however, show that Problem~\eqref{eq:opt:2} is a concave problem (where by concave, we mean that the objective funciton is concave with convex constraint set) by exploiting the special structure of the problem.

\begin{figure}
\begin{centering}
\begin{tikzpicture}[line width=1.5pt,scale=1]

\draw (2,0) arc (0:90:2) ;
\draw[thin, ->] (0,0)--(3.25,0) node[right]{$r_1$}; 
\draw[thin, ->] (0,0)--(0,3.0) node[above]{$r_2$};
\node (A) at (2, -0.25) {$\log(1+P)$};
\node (B) at (-0.8, 2) {$\log(1+P)$};
\draw[fill] (1.55,1.25) circle(0.05cm) node (E){};
\node at (E)[right]  {$\bigl( \log(1+\alpha P), \log(1+(1-\alpha) P)\bigr)$};

\end{tikzpicture}
\caption{Rate region for out-degree $2$ with total power $P$.~\label{fig:rateregion}}
\end{centering}
\end{figure}
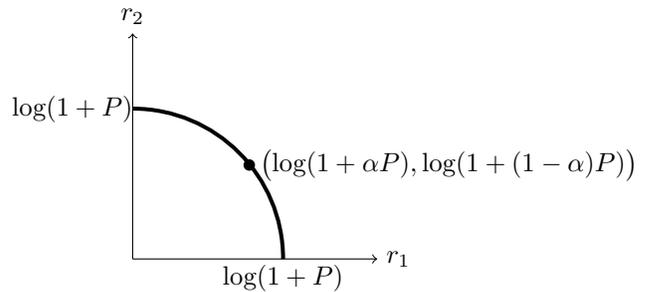

\begin{lemma}\label{lem:concave} Problem \eqref{eq:opt:2} is concave in the underlying variables 
$P_{kl}, l \in O_k, k \in V$.
\end{lemma}

\begin{proof}
Consider two feasible sets of power 
allocation schemes, say $\bar P=\{P_{kl}, k \in V, l \in O_k\}$ and 
$\bar Q=\{Q_{kl}, k \in V, l \in O_k\}$, both respecting the power constraints. 
The former allocates $P_{kl}$ for the edge $(k,l)$ whereas the latter assigns $Q_{kl}$.
For $0\leq \lambda \leq 1$, we have $\lambda \log(1+P_{kl}) + (1-\lambda)\log(1+Q_{kl})$
\begin{multline} \label{eq:obj:concave} \leq 
\log(1 + \lambda  P_{kl} + (1-\lambda)  Q_{kl}). 
\end{multline}
In other words, any linear combination of the rates achieved by the allocations 
$P_{kl}$ and $Q_{kl}$ for the respective fractions of time $\lambda$  and $1-\lambda$
on a link can also be achieved by using a constant power $\lambda P_{kl} + (1-\lambda)Q_{kl}$
for the whole duration. Since the available link rate
got augmented by \eqref{eq:obj:concave}, we know that the solution to max-flow is at 
least as much as the linear combination of the end-to-end flow achieved by $\bar P$ and 
$\bar Q$. This shows the required concavity.
\end{proof}

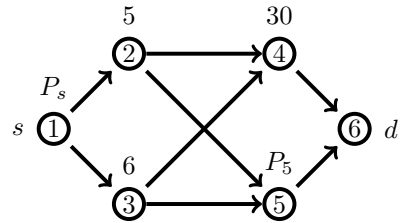
\begin{figure}[h]
\begin{center}
\begin{tikzpicture}[line width=1.5pt, every line/.style={->}, scale=1,
lab/.style={yshift=0.25cm}]
\coordinate (c1) at (0,0); \coordinate (c2) at (1,1); \coordinate (c3) at (1,-1);
\coordinate (c4) at (3,1); \coordinate (c5) at (3,-1); \coordinate (c6) at (4,0);
\foreach \i in {1,...,6}{
\draw[fill=white] (c\i) circle(0.2cm) node (n\i) {$\i$};
};
\draw (n1) edge[->] (n2) (n1) edge[->] (n3) (n2) edge[->] (n4) (n2) edge[->] (n5) 
	(n3)edge[->](n4) (n3)edge[->](n5) (n4) edge[->] (n6) (n5)edge[->](n6);

\draw (n1) node[above,lab]{$P_s$} (n2) node[above,lab]{$5$} (n3) node[above,lab]{$6$}
 	(n4) node[above,lab]{$30$} (n5) node[above,lab]{$P_5$} 
	(n6) node[right, xshift=0.25cm]{$d$} (n1) node[left, xshift=-0.25cm]{$s$};
\end{tikzpicture}
\end{center}
\caption{Fully connected layered network \label{fig:eg:nw}}
\end{figure}

To illustrate the behaviour of the max-flow problem with non-polymatroidal constraints, 
we now consider two examples. The first network, shown in Fig.~\ref{fig:eg:nw},
is a layered network with $2$ layers. The available power
is marked above each node. In order to demonstrate the effect of link capacities on the
max-flow, we will plot the $s\rightarrow d$ max-flow (throughput) as a function of the 
link power (equivalently link rate) $P_{5}$ from node~$5$ to the destination, for 
two different power constraints at the source.

In Fig.  \ref{fig:sim1}, the upper curve (colored blue) shown is for a source power $P_s=20$, while the lower curve (colored red) is for $P_s=15$. When the power 
$P_{5}$ at node~$5$ is low, the source is a power surplus node, i.e., $P_s=15$ or $P_s=20$ gives the same max-flow.  However, as $P_{5}$ increases, the source 
can utilize all its power to increase the max-flow. The solutions for this example
were obtained using
standard numerical solvers from convex programming.

\begin{figure}\label{fig:sim1}
\begin{center}
\begin{tikzpicture}
\begin{axis}[xlabel={$P_{5}$ Watts}, ylabel={Maxflow (b/s/Hz)}, legend style={at={(0.9,0.25)}}, scale=.8 ]
\plot coordinates {
(9.5,6.78463) (8.79685,6.78128) (8.14078,6.77167) (7.52865,6.75636) (6.95751,6.73587) (6.42462,6.71065) (5.92742,6.68109) (5.46351,6.64753) (5.03067,6.61029) (4.62681,6.56965) (4.25,6.52584) (3.89842,6.4791) (3.57039,6.42963) (3.26433,6.37761) (2.97876,6.32321) (2.71231,6.26659) (2.46371,6.20788) (2.23175,6.1472) (2.01533,6.08469) (1.81341,6.02044) (1.625,5.95456) (1.44921,5.88714) (1.2852,5.81826) (1.13216,5.74801) (0.989378,5.67645) (0.856155,5.60367) (0.731854,5.52971) (0.615877,5.45465) (0.507667,5.37853) (0.406703,5.30141) (0.3125,5.22335) (0.224606,5.14438) (0.142598,5.0503)
};
\plot coordinates{
(9.5,6.17493) (8.79685,6.17493) (8.14078,6.17493) (7.52865,6.17493) (6.95751,6.17493) (6.42462,6.17493) (5.92742,6.17493) (5.46351,6.17493) (5.03067,6.17493) (4.62681,6.17493) (4.25,6.17493) (3.89842,6.17493) (3.57039,6.17493) (3.26433,6.17493) (2.97876,6.17493) (2.71231,6.17493) (2.46371,6.1733) (2.23175,6.14679) (2.01533,6.08469) (1.81341,6.02044) (1.625,5.95456) (1.44921,5.88714) (1.2852,5.81826) (1.13216,5.74801) (0.989378,5.67645) (0.856155,5.60367) (0.731854,5.52971) (0.615877,5.45465) (0.507667,5.37853) (0.406703,5.30141) (0.3125,5.22335) (0.224606,5.14438) (0.142598,5.0503)
};
\legend{$P_s=20$,$P_s=15$}
\end{axis}
\end{tikzpicture}
\caption{Max-flow as a function of link power $P_{5}$ \label{fig:eg:rate}}
\end{center}
\end{figure}
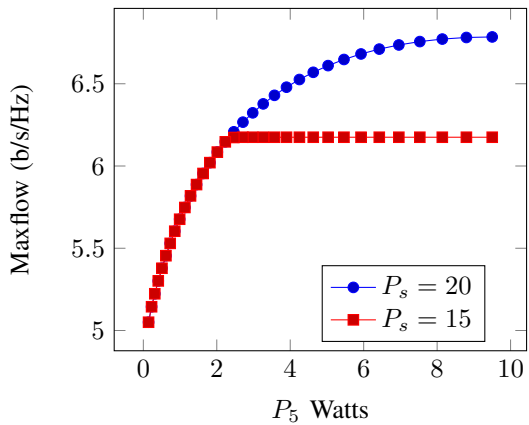
While having concavity is desired, standard numerical solutions will face the curse of dimensionality
when there are many nodes.  Hence our next goal is to identify and exploit sub-structures of 
the problem, where iteratively solving `smaller' problems can lead to global optimal solution
 similar to the classical or  the polymatroidal max-flow problems.

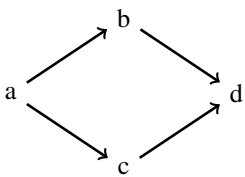
\begin{figure}[H]
\begin{centering}
\begin{tikzpicture}[node distance=2cm,scale=1, line width=1.0pt]
\node (A) at (3, 0) {a};
\node (B) at (4.5, 1) {b};
\node (C) at (4.5, -1) {c};
\node (D) at (6, 0) {d};

\draw[->]
  (A) edge (B) (C) edge (D) (A) edge (C) (B) edge (D)  ;

\end{tikzpicture}
\caption{Example where max-flow is not equal to min-cut.~\label{fig:mincut}}
\end{centering}
\begin{rem}Problem \eqref{eq:opt:2} is also special in comparison to the classical or the polymatroidal max-flow in the sense that the max-flow is not equal to the min-cut. To see this, consider Fig. \ref{fig:mincut}, where the nodes $a,b$ and $c$ have respective powers $P_a, P_b$ and $P_c$, 
with $P_b < < P_a < < P_c$. Consider a cut by the set of links ${(a,b)-(c,d)}$ that separates the nodes $\{a,b\}$ from $\{c,d\}$ (which is also the min-cut), whose cut-capacity is $\log(1+P_a) + \log(1+P_b)$. However, the achievable max-flow is only $\log(1+\alpha P_a) + \log(1+P_b)$, where $\alpha$ is such that $(1-\alpha)P_a = P_b$ since $(1-\alpha)P_a$ amount of power is routed from $a$ to $b$ to completely utilize the capability of node $b$. In case of polymatroidal constraints, in this example, the min-cut will be $(a) - (b,c)$, and the min-cut capacity will be equal to the max-flow \cite{Lawler}.
\end{rem}

\end{figure}

Next, to highlight the basic idea on how to solve Problem \eqref{eq:opt:2}, we consider a $2$-layer network, and propose an algorithm which performs iterative rate optimization only between layers. The optimality of the algorithm will also be shown.
 Extension to more than $2$-layers is described in 
Section~\ref{sec:mult:layer}.

%% file: RateMax3.tex
\section{Optimal Algorithm for $K=2$-Layer Network} \label{sec:two:layer}

Recall that in a layered network, the 
source connects to all nodes in the first intermediate layer $\mathcal L_1$. 
Since $K=2$, each path from source to destination
has $3$ hops, this is illustrated for an example network in Fig.~\ref{fig:eg:nw}. All nodes in the set 
$\mathcal L_2$ connect to the destination. The middle section 
comprises an arbitrary subgraph of edges between $\mathcal L_1$ and  $\mathcal L_2$. 
Since the out-degree of all nodes in layer $2$ towards the destination is $1$, the power allocation for these nodes towards the destination is trivial, and we define their achievable rates as  $m_{jd} = \log(1+P_{jd})$ for $j \in \cL_2$.
We propose the following algorithm to find the optimal flow (refer Problem~\eqref{eq:opt:2}) from the source to the destination for $K=2$.

\begin{figure}[H]

\hrule

\vspace*{0.15cm}
\noindent ALGORITHM~FlowMax

\vspace*{0.15cm}
\hrule

\vspace*{0.25cm}
\noindent \textbf{Step I:} Initially set $P_{si} = \frac{P_s}{|\mathcal L_1|}$ $\forall i\in \cL_1$, and $m_{jd} = \log(1+P_{jd}),\forall j \in \cL_2$. \\ Initialize, counter $\sfc=1$, $R(0) = 0$.

\noindent \textbf{Step II:} Assign $f_i = \log (1 + P_{si})$ for $i\in \cL_1$.

\noindent \textbf{Step III:} For nodes $i \in \mathcal L_{1}$, find the optimal outgoing rates $r_{i}=\sum_{l \in O_i} r_{il}$ by solving
$$
R(\sfc) = \max \sum_{i\in \mathcal L_1} r_{i} \text{ such that } 
$$
 $ r_{i} \leq f_i \ \text{and} \  \sum_{l \in I_j}r_{lj}  \le m_{jd}, j \in \cL_2$.\\

  

\noindent \textbf{Step IV:} Define $U =\{i \in \cL_1: r_{i} < f_i\}$ \\
{\bf If $\Bigl( |U|=0$ or $|U|=|\cL_1|$ or $|R(\sfc) - R(\sfc-1)| \le \epsilon\Bigr)$} \textbf{break;}

\noindent{\textbf{Else}} \\
Compute the effective unused source power as
$$
\Delta =   \left( \sum_{i \in U} (P_{si} - e^{r_{i}} + 1) \right).
$$
Redistribute the unused power as
\begin{itemize}
\item for each $j \in U^c$, \ \ 
$
P_{sj} = P_{sj} +  \frac{\Delta}{|\mathcal L_1|}.
$ 
\item for each $i \in U\phantom{^c}$, \ \ 
$
P_{si} = e^{r_{i}} - 1 + \frac{\Delta}{|\mathcal L_1|}.
$ 
\end{itemize}
$\sfc=\sfc+1$, Go back to Step~II \\
{\bf EndIf}

\vspace*{0.1cm}
\hrule
\end{figure}

The main idea of the algorithm is to initially assign equal power from the source to all its outgoing edges in \textbf{Step II}. 
With equal power allocation, let $f_i,i\in \cL_1$ be the incoming rate into node $i \in \cL_1$ from the source. Subject to incoming constraints $f_i$ for nodes $i \in \cL_1$ and out-going rate constraints of $m_{jd}$ for nodes $j\in\cL_2$ to the destination, in \textbf{Step III}, we find the optimal sum-rate between nodes of layer $1$ and $2$, where the optimal out-going rate for nodes $i \in \cL_1$ is denoted by $r_{i}$.

The collection of nodes $i \in \cL_1$ for which the out-rate $r_{i}$ computed in \textbf{Step III} is lower than the incoming rate $f_i$ they are receiving from the source is called $U$. Nodes in $U$ are unable to support the rate they are getting in from the source. In the next iteration, power from the source is reduced towards nodes of $U$ and increased towards $U^c$ to update $f_i$, i.e., $f_i$ is increased for nodes $i \in U^c$ and decreased for $i\in U$. One important point is that even after updation of $f_i$'s we do not make $f_i = r_{i}$ but instead keep $f_i > r_{i}$. This might slow the algorithm's speed, however, avoids technical difficulty in proving its optimality. We show in Lemmas \ref{lem:U} and \ref{lem:nonvoidU} 
that if in any iteration $|U|= |\cL_1|$ or $|U|=0$ (in which case the algorithm terminates), respectively, then the corresponding rates obtained are optimal. Otherwise, the algorithm terminates at convergence.



\begin{theorem}\label{thm:main}
For $K=2$, Algorithm~FlowMax  converges to the optimal solution of Problem \eqref{eq:opt:2}.
\end{theorem}
\begin{proof} 
Following Lemma \ref{lem:nondec}, Lemma \ref{lem:U}, Lemma \ref{lem:nonvoidU}, Lemma \ref{lem:localopt}, it follows that whenever the algorithm \text{FlowMax} stops (break condition is satisfied), a rate arbitrarily close to the optimal rate is achieved (specified by choosing any $\epsilon >0$).
\end{proof}

\begin{lemma}\label{lem:nondec} 
The sum-rate $R(\sfc)$ computed from layer $1$ to $2$ is non-decreasing in $c$.
\end{lemma}
\begin{proof} From the definition of the algorithm,  $f_{j}(\sfc+1) > f_{j}(\sfc)$ for $j \in U^c(\sfc)$ and 
$f_{j}(\sfc+1) > r_{j}(\sfc)$ for $j \in U(\sfc)$. Thus, in each iteration, the effective constraints $f_j$ for sum-rate maximization in Step \textbf{ III}  are strictly enlarged.
\end{proof}

%
\begin{lemma}\label{lem:U} If in iteration $\sfc$, $U(\sfc) = \mathcal L_1$, then $R(\sfc)$ is optimal.
\end{lemma}
\begin{proof}
If $U(\sfc) = \mathcal L_1 $, then $r_{i}(\sfc) < f_i(\sfc)$ for all $i \in \mathcal L_1$, 
where $r_{i}(\sfc)$ is the optimal rate computed by the optimal 
sum-rate algorithm for node $i\in \cL_1$ in Step \textbf{ III}. Thus, $R(\sfc) = \sum_{i\in\cL_1}r_{i1}$ is an upper bound on the achievable rate. 
This is also achievable by just reducing the rate from source to node $i$ 
from $f_i(\sfc)$ (achievable from Step 1 of iteration $t$) to $r_{i}(\sfc)$.
\end{proof}
\begin{lemma}\label{lem:nonvoidU} We have $\forall \ \sfc \geq 2, |U(\sfc)| > 0$. Furthermore,  $|U(1)| = 0$
will imply that $R(1)$ is the optimal throughput.
\end{lemma}
\begin{proof}
The second statement is proved first. Notice that we started with equal power allocation from source to define
$f_i = \log(1+ \frac{P_s}{|\cL_1|}), \forall i \in \mathcal \cL_1$. 
Thus as discussed before, due to the concavity of the logarithm, 
$m_s := \sum_{i=1}^{|\cL_1|} f_i$  is the largest rate 
the source can transmit at. If $|U(1)| = 0$, this means that 
$r_{i} = f_i$ is achievable for all $i \in \mathcal L_1$, and hence 
$\sum_{i=1}^{|\cL_1|}r_{i} = \sum_{i=1}^{|\cL_1|}f_i = m_s$, the maximal  throughput from
source to $\mathcal L_1$, is achieved.

For the first statement of Lemma~\ref{lem:nonvoidU}, let $\sfc$ be the earliest iteration where $|U(\sfc)| = 0$ for $\sfc>1$. 
Thus, in iteration $\sfc$, $r_{i}(\sfc) = f_i(\sfc)$ for all nodes $i$ of layer $1$. For the set $U(\sfc-1)$
at iteration~$\sfc-1$, we claim that if $|U(\sfc)| = 0$ for $\sfc>1$, then $|U(\sfc-1)|=0$ as well, 
contradicting the existence of an earliest such instant for $\sfc>1$.
If $|U(\sfc-1)|>0$, then going from iteration $\sfc-1$ to $\sfc$, the constraint $f_i(\sfc-1)$ for 
$i \in U^c(\sfc-1)$ is relaxed to $f_i(\sfc) > f_i(\sfc-1)$, and  our algorithm also ensures 
$f_i(\sfc) > r_{i}(\sfc-1)$ for $i \in U(\sfc-1)$. 

Now  $U(\sfc) = 0$ will imply that the rates $r_i$ got increased for all $i \in \mathcal L_1$,
while
%
going from iteration $\sfc-1$ to $\sfc$. Thus, a larger sum-rate is feasible for nodes of $i \in U(\sfc-1)$ in iteration $\sfc-1$ without decreasing the rate for nodes of $i \in U^c(\sfc-1)$, contradicting the optimality of rate vector $[r_1(\sfc-1) \dots r_{|\mathcal L_1|}(\sfc-1)]$ found in Step \textbf{ III} of iteration $\sfc-1$.
\end{proof}
\begin{lemma}\label{lem:localopt} If the sum-rate satisfies $R(\sfc+1) = R(\sfc)$ for any iteration $\sfc$,
then the rate vector $\br(\sfc) = [r_1(\sfc) \dots r_{|\mathcal L_1|}(\sfc)]$ is a global maxima. 
\begin{proof} From the definition of the algorithm,  $f_{j}(\sfc+1) > f_{j}(\sfc)$ for $j \in U^c(\sfc)$ and 
$f_{j}(\sfc+1) > r_{j}(\sfc)$ for $j \in U(\sfc)$. Thus, if 
$R(\sfc+1) = R(\sfc)$, that means that in the strictly open neighborhood of $\br(\sfc)$, there is no ascent direction available. Since Problem \eqref{eq:opt:2} is concave, it follows that the rate output by algorithm 
\text{FlowMax} is in fact optimal if $R(\sfc+1) = R(\sfc)$.
\end{proof}
\end{lemma}

\subsection{Non-orthogonal links}\label{sec:lay}
Recall that our network model assumed non-interfering or orthogonal links at each node. However, Algorithm~FlowMax can also accommodate interfering links at the receivers. For example, if each node has multiple access constraints \cite{tse1998multiaccess} on its incoming edges, then the incoming rate constraints are polymatroidal, and we can extend our results for layered networks. Recall that we earlier argued that without any constraints on incoming edges of any node DAG is equivalent to a layered network. This assertion need not be true while incorporating interfering links. 
\begin{lemma} For a layered network (not necessarily a DAG) with $K=2$, Algorithm~FlowMax  converges to the optimal solution of Problem~\eqref{eq:opt:2} even when additional 
receiver side polymatroidal constraints are imposed on the incoming edges to any node.
\end{lemma}
The proof is essentially the same as that of Theorem \ref{thm:main}, since even with polymatroidal constraints enforced at nodes of layer $2$, the sum-rate maximization between layer $1$ and layer $2$ for each iteration is a concave problem as before. This result can be extended for any number of layers $K$, similar to the orthogonal links case in Section~\ref{sec:mult:layer}.

%% file: multlayerVer1.tex

\section{Multi-layer Network with $K > 2$} \label{sec:mult:layer}

\begin{figure}
\begin{centering}
\begin{tikzpicture}[node distance=2cm,scale=1]
\node (A) at (3, 0) {a};
\node (B) at (4.5, 1) {b};
\node (C) at (4.5, 0) {c};
\node (E) at (4.5, 2) {e};
\node (F) at (3, 1) {f};
\draw (C) circle(0.2cm);

\node (A') at (6, 0) {a};
\node (B') at (7.5, 1) {b};
\node (C') at (7.5, 0) {c};
\node (E') at (7.5, 2) {e};
\node (F') at (6, 1) {f};
\draw (C') circle(0.2cm);
\draw (B') circle(0.2cm);

\draw[->]
  (A) edge (B)  (A) edge (C)  (F) edge (B) (F) edge (E) ;
  \draw[->]
  (A') edge (B')  (A') edge (C')  (F') edge (B') (F') edge (E') ;

\end{tikzpicture}
\caption{Illustration of power reallocation in PowerAug for a layer connected network.~\label{fig:power}}
\end{centering}
\end{figure}
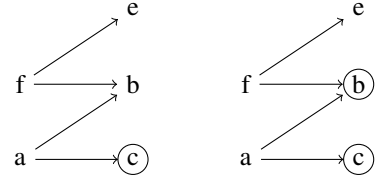

In this section, we generalize the algorithm FlowMax to  $K>2$ layers. For ease of exposition,  we only describe the algorithm, when the network is layer-connected, i.e., any two nodes $k_1,k_2$ of layer $j$ are reachable from each other by using only edges (without considering the direction) between layer $j-1$ and $j$, e.g. see Fig. \ref{fig:power}, where nodes $c, b,$ and $e$ are reachable from each other using edges from only the preceding layer. We omit the details for a general network where to reach node $k_1$ from $k_2$ both of layer $j$, we might have to reach nodes of layers $j-2$ or lower (worst case,  the source) as shown in Fig. \ref{fig:power2}.

Let $P_{uv}$ be the power assigned on the directed edge which connects node $u$ to $v$ of two consecutive layers,
and $r_{uv}:= \log(1+P_{uv})$. We will take $P_{us}$ (for source $s$)
to be infinity, as there is no constraint before the source node. Let $d_u$ denote the
outgoing degree of node~$u$. 
For a set $S \subset \cL_\ell$, let 
$\cN(S)$ be the set of nodes in the preceding layer $\cL_{\ell-1}$, having an edge to some node of $S$, i.e., $\cN(S) =\{j: j\in \cL_{\ell-1}, e(j,v) = 1, \ \text{for some} \ v \in S\}.$
\def\ww{\phantom{ww}}
\begin{figure}[H]

\hrule

\vspace*{0.15cm}
\noindent ALGORITHM~FlowMax-II

\vspace*{0.15cm}
\hrule

\noindent \textbf{Step I:} 
Let $P_{iv} =  \frac{P_i}{d_i}, \forall v \in O_i$, and  $P_{js} = \infty$ for the source node~$s$. Start
with layer~$k=2$. 
\\ Set count $\sfc=1$. \\

\noindent \textbf{Step II:} (Vector $\{r_i, i \in \cL_{k-1}\}, \cP) =$ LayerOPT($k, \cP$).  \\
$$
\Delta = \sum_{i\in \cL_{k-1}} {f_i} - \sum_{i \in \cL_{k-1}} r_i.
$$ 
\noindent \textbf{Step III:} Declare $U(\sfc) =\{i\in \cL_{k-1}: r_i < f_i \}$. \\
\textbf{If} $\bigl((U(\sfc)=\emptyset) \text{ OR }  (U(\sfc)=|\cL_{k-1}|) \text{ OR } (k=1) \text{ OR } (\Delta \leq \epsilon) \bigr)$ \\
\ww $k=k+1$.\\
\textbf{Else} \\
\ww Update Powers: $\cP  = \text{PowerAug}(k, U(\sfc), \cP)$. \\
\ww For $\ell=k-2$ to $2$, run LayerOPT($\ell,\cP$), sequentially. \\
Set $k=1$.
\textbf{EndIf} \\
\noindent \textbf{Step IV:} $\sfc=\sfc+1$, go to Step~II if $k < K$. If $k=K$, break.

\hrule

\end{figure}

\begin{figure}
\begin{centering}
\begin{tikzpicture}[node distance=2cm,scale=1]
\node (A) at (3, 0) {a};
\node (B) at (4, 1) {b};
\node (C) at (4, 0) {c};
\node (E) at (4, 2) {e};
\node (F) at (3, 1) {f};
\draw (C) circle(0.2cm);
\node (G)  at (3,2){g};
\node (H)  at (2,2){h};
\node (I)  at (2,1){i};


\node (A') at (6, 0) {a};
\node (B') at (7, 1) {b};
\node (C') at (7, 0) {c};
\node (E') at (7, 2) {e};
\node (F') at (6, 1) {f};
\draw (C') circle(0.2cm);
\draw (B') circle(0.2cm);
\draw (F') circle(0.2cm);
\node (G')  at (6,2){g};
\node (H')  at (5,2){h};
\node (I')  at (5,1){i};

\draw[->]
  (A) edge (B)  (A) edge (C)  (F) edge (B) (G) edge (E)  (I) edge (A)  (I) edge (F) (I) edge (G) (H) edge (G) ;
  \draw[->]
  (A') edge (B')  (A') edge (C')  (F') edge (B') (G') edge (E') (I') edge (A')  (I') edge (F') (I') edge (G') (H') edge (G');

\end{tikzpicture}
\caption{Illustration of power reallocation in PowerAug for a non-layer connected network.~\label{fig:power2}}
\end{centering}
\end{figure}
\begin{figure}[H]

\hrule

\vspace*{0.15cm}
\noindent SUBROUTINE~LayerOPT($k,\cP$)

\vspace*{0.15cm}
\hrule

\noindent Input: Layer~$k$, Power allocation~$\cP$ of all links.\\
\noindent Output: Rate vector $r_i, i\in \cL_{k-1}$ and associated new power allocation~$\cP$ of all links.\\
\noindent \textbf{Step I:} Assign  
\vspace{-0.1in}
\begin{align*}
f_i &= \sum_{u \in I_i} \log (1 + P_{ui}), \forall i \in \cL_{k-1} \\
g_j &= \sum_{v \in O_j} \log(1 + P_{jv}),  \forall j \in \cL_{k}.
\end{align*}

\noindent \textbf{Step II:}  Find optimal $r_i = \sum_{j\in O_i} r_{ij}$ (sum-rate out of node $i$ in layer $k-1$, where $r_{ij}= \log(1+P_{ij}), \forall i \in \cL_{k-1}$) by solving
\begin{align} \label{eq:layer:opt}
\max \sum_{i\in \mathcal L_{k-1}} r_i \text{ such that } r_i \leq f_i \text{ and } \sum_{i \in I_j} r_{ij} \leq g_j.
\end{align}

\hrule
\end{figure}
Following the same philosophy of ALGORITHM~FlowMax for the $K=2$-layer network,  ALGORITHM~FlowMax-II for a multi-layer connected network (MLN) is proposed where the main idea is to sequentially solve a sub-problem (LayerOPT($k,\cP$)), that is a sum-rate maximization problem between layer~$k-1$ and layer~$k$ for some $k$, with incoming flow constraints $f_i, i\in \cL_{k-1}$ from layer $k-2$ and outgoing flow constraints $g_j$ for nodes $j$ of layer $k$. 

Starting with $k=2$, solving LayerOPT($k,\cP$) gives rate $r_i$ for $i\in\cL_{k-1}$. If the sum-rate ($\sum_{i\in\cL_{k-1}}r_i$) output by LayerOPT($k,\cP$) is almost (additive difference of $\epsilon$) equal to the sum of incoming flow constraints ($\sum_{i\in\cL_{k-1}}f_i$) from layer $k-2$, we move to the next layer, and solve LayerOPT($k+1,\cP$). 
Otherwise, we need to reduce (increase) the rate coming into nodes $i\in \cL_{K-1}$ for which $r_i<f_i$ ($r_i=f_i$) by reallocating the power on outgoing links of nodes of layer $k-2$, as done in subroutine PowerAug. This power reallocation, changes the outgoing flow constraints for nodes in layer $k-2$, and subsequently to maintain feasibility, we find LayerOPT($j,\cP$) iteratively for $j=k-2$ till layer $2$ subject to the new outgoing rate constraints from nodes of layer $k-2$.

\begin{figure}

\hrule
\noindent SUBROUTINE~PowerAug
\hrule 

Input = ($k, U, \cP = \{P_{uv}\}$): Layer $k$, set $U$ of layer $k-1$, 
current power allocation $\cP$ for all nodes. \\
For a set $S \subset \cL_\ell$, let \\
$\cN(S) =\{j: j\in \cL_{\ell-1}, e(j,v) = 1, \ \text{for some} \ v \in S\}$, \\ i.e., the set of nodes of 
the preceding layer than have an edge to some node of $S$.
For $v\in U$ of layer $k-1$, \\ 
let $r_v = \sum_{i\in \mathcal L_{k}} r_{vi}$ and $f_v = \sum_{t\in \mathcal L_{k-2}} \log(1+P_{tv})$, \text{where} \ $f_v > r_v$ since $v \in U$.\\
{\bf While}{ there exists any $u \in \cN(U) \cap \cN(U^c)$ of layer $k-2$}, {\bf do}\\
Pick any $v\in U$\\ 
{\bf If}{ $\sum_{t\in \mathcal L_{k-2} \backslash u} \log(1+P_{tv}) > r_v
$} \% incoming rate into $v$ from nodes other than $u$ is $> r_v$ \%\\
\ \ \ \ \ $P_{uv} = \frac{P_{uv}}{2}$, \\
\ \ \ \ \  $P_{uw} = P_{uw} + \frac{P_{uv}}{2d_u}, \ \\
 \forall \  w \in U^c$ such that $(u,w)$ is an edge ;\\
where $d_u$ is the out-degree of $u$ with edges in $U^c$\\
{\bf Else}\\ 
Let $\log(1+{\tilde P}) = r_v - \sum_{t\in \cL_{k-2}\backslash u} \log(1+P_{tv})$\\
$P_{uv} = {\tilde P} + \frac{P_{uv}-{\tilde P}}{1+d_u}$,\\
$P_{uw} = P_{uw} + \frac{P_{uv}-{\tilde P}}{1+d_u}, \  \forall \  w \in U^c$ such that $(u,w)$ is an edge;\\
{\bf End If}\\  
Update $U = U \cup \{w \in U^c\}$ for $w$ such that $(u,w)$ is an edge \\ 
{\bf End While} \\
\hrule
\end{figure}

With the layer connected network assumption, the power reallocation in subroutine PowerAug is done as follows. In first iteration, set of nodes $i \in \cL_{k-1}$  with $r_i<f_i$ ($r_i=f_i$) are called $U(0)$ $(U^c(0))$. 
We find a node of layer $k-2$ that has an edge to both sets $U(0)$ and $U^c(0)$ and decrease the power on the link towards $U(0)$ and increase it on all links of $U^c(0)$, such that for nodes in $U(0)$ even after updation $f_i > r_i$. Then we include the nodes of layer $l-1$ for which power has been increased on at least one incoming link, into set $U(0)$ (remove it from $U^c(0)$) and call it $U(1)$. Repeat the above process until there is any node in $U^c(0)$.  For example, see Fig. \ref{fig:power} where the considered two layers are  layer-connected, and only node $c$ (circled) is part of $U(0)$. Hence, power is decreased on link $(a,c)$ and increased on $(a,b)$. Subsequently, node $b$ is made part of $U(1)$ (circled) and power is decreased on link $(f,b)$ and increased on $(f,e)$. Since the network is layer connected, it is easy to see that at the end of this procedure, all nodes of layer $k-1$ have their $f_i$ increased for $i \in U^c(0)$ and $f_i > r_i$ (by choice) for $i\in U(0)$.

When the network is not layer-connected as shown in Fig. \ref{fig:power2}, the power is reallocated via node i and not directly via node f.

\begin{theorem}\label{thm:maingen}
Algorithm~FlowMax-II converges to the optimal solution of Problem \eqref{eq:opt:2} if the network is layer connected.
\end{theorem}
\begin{proof}
If Algorithm~FlowMax-II never encounters the Else condition in Step III, i.e., it never encounters a bottleneck layer and power allocations on previous layers need not be updated, then the optimality is obvious from Lemma \ref{lem:U}, Lemma \ref{lem:nonvoidU}, and Lemma \ref{lem:localopt}.
If Algorithm~FlowMax-II does encounter the Else condition in Step III for some iteration, then in Lemma \ref{lem:multifinal}, we show via Lemma \ref{lem:multi1}, \ref{lem:multi2}, and, \ref{lem:multi3} that 
 the minimum achievable intra-layer sum-rate ($\min_k \sum_{i\in\cL_{k}}r_i$) is non-decreasing in any iteration. Eventually, the Else condition in Step III will not be encountered for any layer, and the optimality will follow from Lemma \ref{lem:U}, Lemma \ref{lem:nonvoidU}, and Lemma \ref{lem:localopt}.
\end{proof}
\begin{definition} Let the sum-rate ($\sum_{i \in \cL_{k-1}} r_i$) between layer $k-1$ to layer $k$ be defined as $R_{k-1}$. With the Algorithm~FlowMax-II, sum-rates $R_{i}, i=1, \dots, K$ are updated sequentially from left to right, and then right to left whenever Else condition in Step III is encountered (a bottleneck layer). To distinguish between left to right and right to left updates, 
we define $\overrightarrow{R}_k$ and $\overleftarrow{R}_k$ as the rate achieved between layer $k-1$ to layer $k$ on the left to right (forward) and right to left (backward) 
iterations, respectively. 
\end{definition}
\begin{lemma}\label{lem:multi1} Let the Algorithm~FlowMax-II be working on layer $\cL_b$ and satisfy the Else condition in Step III, i.e., it has hit a bottleneck and power allocations on previous layers needs to be updated. Let the current sum rate from layer $\cL_b$ to $\cL_{b+1}$ be $R_{b}$. Then after power augmentation (by subroutine \text{PowerAug}) on outgoing edges from layer $\cL_{b-1}$ towards layer $\cL_b$, the sum-rate from layer $\cL_{b-1}$ to $\cL_{b}$, defined $R_{b-1}$ is at least as much as $R_{b}$. 
\end{lemma}
\begin{proof}Subroutine \text{PowerAug} ensures that even after power augmentation on outgoing edges from layer $\cL_{b-1}$ towards layer $\cL_b$, the rate $R_{b}$ is achievable from $\cL_b$ to $\cL_{b+1}$, i.e., after power augmentation, the incoming sum-rate from layer $\cL_{b-1}$ to layer $\cL_b$, $\sum_{i\in \cL_b} f_i \ge R_b$. Since $R_{b-1} = \sum_{i\in \cL_b} f_i$, we have $R_{b-1} \ge R_b$.
\end{proof} 

\begin{lemma}\label{lem:multi2} Let the Algorithm~FlowMax-II be working on layer $\cL_b$ and satisfy the Else condition in Step III. Let the current sum rate from layer $\cL_b$ to $\cL_{b+1}$ be $R_{b}$. Then when subroutine LayerOPT is run for layers $\cL_{b-2}$ till layer $2$ from right to left, consecutively, the sum-rate obtained on layer $\cL_k$, $2\le k \le   {b-2}$ (defined as $\overleftarrow{R}_k$) is at least as much as $R_b$. 
\end{lemma}
\begin{proof} It is important to note that when Algorithm~FlowMax-II is working on layer $\cL_b$ and satisfy the Else condition in Step III, then the current sum-rate $\overrightarrow{R}_k$ on all layers $\cL_k, k < b$, (found for layers from left to right until previous iteration) satisfies $\overrightarrow{R}_k \ge R_{b}$, since layer $b$ is the current bottleneck. Since a larger rate than $R_{b}$ is achievable on all previous layers, even after power augmentation, to change the  sum-rate from layer $\cL_{b-1}$ to layer $\cL_b$, by continuity, a rate larger than $R_{b}$ is still achievable on previous layers, implying that $\overleftarrow{R}_k \ge R_{b}$.  
\end{proof}

Note that it is possible that 
$\overleftarrow{R}_k \le \overrightarrow{R}_k$ in consecutive updates, but we only need that 
$\overleftarrow{R}_k \ge \min_{1\le \ell \le b} R_\ell = R_b$ for all $k \le b$.
\begin{lemma}\label{lem:multi3} Let the Algorithm~FlowMax-II be working on layer $\cL_b$ and satisfy the Else condition in Step III. Let the current sum rate from layer $\cL_b$ to $\cL_{b+1}$ be $R_{b}$. Let the subroutine LayerOPT has been run for layers $\cL_{b-2}$ till layer $2$ from right to left, and $\overleftarrow{R}_k$ has been found. Then when the Algorithm~FlowMax restarts going from left to right, let the  sum-rate in layer $k$ be defined as $\overrightarrow{R}_k$. Then $\overrightarrow{R}_k \ge \overleftarrow{R}_k$ for each $k$ till layer $b$.
\end{lemma}
\begin{proof} Follows from Lemma \ref{lem:nondec}.
\end{proof}
 
 \begin{lemma}\label{lem:multifinal} The bottleneck layer rate $\min_{1\le \ell \le K} R_\ell$ is non-decreasing in each iteration of Algorithm~FlowMax-II.
 \end{lemma}
 \begin{proof} We know that whenever a bottleneck layer $\cL_b$ is encountered by the algorithm, (Else condition is satisfied in Step III with sum-rate $R_b$),  then one pass from layer $b-2$ to $2$ and one pass from layer $2$ to $b-1$ is made to update $\overleftarrow{R}_k(new)$ and $\overrightarrow{R}_k(new)$, $k\le b-1$.
 
 From Lemma \ref{lem:multi1}, \ref{lem:multi2}, \ref{lem:multi3}, we know that $\overleftarrow{R}_k(new) \ge R_b$ as well as $\overrightarrow{R}_k(new) \ge R_b$. After this, the subroutine \text{LayerOpt} is run for layer $\cL_b$, and the the updated rate $R_b(new)$ is at least as much as before following the same argument as in Lemma \ref{lem:nondec}. Since $R_b$ is the current minimum sum-rate, the result follows.
 \end{proof}

\section{Simulations}\label{sec:sim}
In this section, we illustrate the numerical performance of our algorithm to maximize the max-flow \eqref{eq:opt:2}. We consider the $2$-layer network shown in Fig.~\ref{fig:eg:nw} and plot the max-flow for various values of $P_s$ and $P_5$ obtained via Algorithm~FlowMax. 
Recall that Fig.~\ref{fig:eg:rate} was generated directly by solving Problem \eqref{eq:opt:2}  using a convex solver while Fig. \ref{fig:eg:ratesim} is obtained by executing Algorithm~FlowMax. It is worthwhile to note that corresponding curves for Fig.~\ref{fig:eg:rate} and Fig. \ref{fig:eg:ratesim} exactly match, where  Algorithm~FlowMax algorithm converged in at most $5$ iterations for each value of $P_5$.

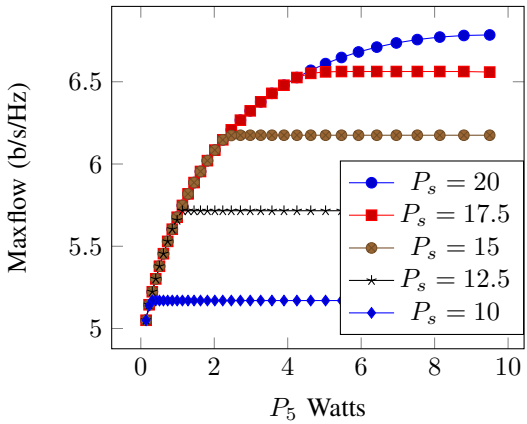
\begin{figure}
\begin{tikzpicture}
\begin{axis}[xlabel={$P_{5}$ Watts}, ylabel={Maxflow (b/s/Hz)}, legend style={at={(1,0.55)}}, scale=.8 ]
\plot coordinates {
(9.5,6.78463) (8.79685,6.78128) (8.14078,6.77167) (7.52865,6.75636) (6.95751,6.73587) (6.42462,6.71065) (5.92742,6.68109) (5.46351,6.64753) (5.03067,6.61029) (4.62681,6.56965) (4.25,6.52584) (3.89842,6.4791) (3.57039,6.42963) (3.26433,6.37761) (2.97876,6.32321) (2.71231,6.26659) (2.46371,6.20788) (2.23175,6.1472) (2.01533,6.08469) (1.81341,6.02044) (1.625,5.95456) (1.44921,5.88714) (1.2852,5.81826) (1.13216,5.74801) (0.989378,5.67645) (0.856155,5.60367) (0.731854,5.52971) (0.615877,5.45465) (0.507667,5.37853) (0.406703,5.30141) (0.3125,5.22335) (0.224606,5.14438) (0.142598,5.0503)
};
\plot coordinates{
(9.5,6.55902) (8.79685,6.56219) (8.14078,6.56224) (7.52865,6.56224) (6.95751,6.56224) (6.42462,6.56224) (5.92742,6.56224) (5.46351,6.56221) (5.03067,6.55997) (4.62681,6.55117) (4.25,6.52573) (3.89842,6.4791) (3.57039,6.42963) (3.26433,6.37761) (2.97876,6.32321) (2.71231,6.26659) (2.46371,6.20788) (2.23175,6.1472) (2.01533,6.08469) (1.81341,6.02044) (1.625,5.95456) (1.44921,5.88714) (1.2852,5.81826) (1.13216,5.74801) (0.989378,5.67645) (0.856155,5.60367) (0.731854,5.52971) (0.615877,5.45465) (0.507667,5.37853) (0.406703,5.30141) (0.3125,5.22335) (0.224606,5.14438) (0.142598,5.0503)
};
\plot coordinates{
(9.5,6.17493) (8.79685,6.17493) (8.14078,6.17493) (7.52865,6.17493) (6.95751,6.17493) (6.42462,6.17493) (5.92742,6.17493) (5.46351,6.17493) (5.03067,6.17493) (4.62681,6.17493) (4.25,6.17493) (3.89842,6.17493) (3.57039,6.17493) (3.26433,6.17493) (2.97876,6.17493) (2.71231,6.17493) (2.46371,6.1733) (2.23175,6.14679) (2.01533,6.08469) (1.81341,6.02044) (1.625,5.95456) (1.44921,5.88714) (1.2852,5.81826) (1.13216,5.74801) (0.989378,5.67645) (0.856155,5.60367) (0.731854,5.52971) (0.615877,5.45465) (0.507667,5.37853) (0.406703,5.30141) (0.3125,5.22335) (0.224606,5.14438) (0.142598,5.0503)
};
\plot coordinates{
(9.5,5.71596) (8.79685,5.71596) (8.14078,5.71596) (7.52865,5.71596) (6.95751,5.71596) (6.42462,5.71596) (5.92742,5.71596) (5.46351,5.71596) (5.03067,5.71596) (4.62681,5.71596) (4.25,5.71596) (3.89842,5.71596) (3.57039,5.71596) (3.26433,5.71596) (2.97876,5.71596) (2.71231,5.71596) (2.46371,5.71596) (2.23175,5.71596) (2.01533,5.71596) (1.81341,5.71596) (1.625,5.71596) (1.44921,5.71596) (1.2852,5.71596) (1.13216,5.71596) (0.989378,5.6584) (0.856155,5.60356) (0.731854,5.52971) (0.615877,5.45465) (0.507667,5.37853) (0.406703,5.30141) (0.3125,5.22335) (0.224606,5.14438) (0.142598,5.0503)
};
\plot coordinates{
(9.5,5.16992) (8.79685,5.16992) (8.14078,5.16992) (7.52865,5.16992) (6.95751,5.16992) (6.42462,5.16992) (5.92742,5.16992) (5.46351,5.16992) (5.03067,5.16992) (4.62681,5.16992) (4.25,5.16992) (3.89842,5.16992) (3.57039,5.16992) (3.26433,5.16992) (2.97876,5.16992) (2.71231,5.16992) (2.46371,5.16992) (2.23175,5.16992) (2.01533,5.16992) (1.81341,5.16992) (1.625,5.16992) (1.44921,5.16992) (1.2852,5.16992) (1.13216,5.16992) (0.989378,5.16992) (0.856155,5.16992) (0.731854,5.16992) (0.615877,5.16993) (0.507667,5.16993) (0.406703,5.16993) (0.3125,5.16993) (0.224606,5.13971) (0.142598,5.0503)
};
\legend{$P_s=20$, $P_s=17.5$,$P_s=15$, $P_s=12.5$, $P_s=10$}
\end{axis}
\end{tikzpicture}
\caption{Max-flow as a function of link power $P_{5}$ \label{fig:eg:ratesim}}
\end{figure}

Next, to model the non-orthogonal links, we once again consider the two-layer network of Fig.~\ref{fig:eg:nw}, and let the edges incident to node $4$ and $5$ have constraints defined by the rate region of a Gaussian multiple access channel, which is polymatroidal. Thus, the change needed in Algorithm~FlowMax is only in Step~III, where additional polymatroidal constraints are imposed on the rates from layer $\mathcal L_1$ to $\mathcal L_2$, without losing out on the concavity of maximization between layer $\mathcal L_1$ and $\mathcal L_2$.  We demonstrate the
throughput performance under additional Gaussian MAC rate constraints 
on nodes $4$ and $5$ in Fig.~\ref{fig:eg:mac}. Fig.~\ref{fig:eg:mac} and Fig.~\ref{fig:eg:ratesim} are comparable for $P_s=20$, and it is worthwhile noting that the max-flow achieved with interfering links is significantly smaller as expected.
 
\begin{figure}[H]
\begin{tikzpicture}
\begin{axis}[xlabel={$P_{5}$ Watts}, ylabel={Maxflow (b/s/Hz)}, legend style={at={(0.9,0.275)}}, scale=1 ]
\plot coordinates{
(9.5,5.16993) (8.79685,5.16993) (8.14078,5.16993) (7.52865,5.16993) (6.95751,5.16993) (6.42462,5.16993) (5.92742,5.16993) (5.46351,5.16993) (5.03067,5.16993) (4.62681,5.16993) (4.25,5.16993) (3.89842,5.16993) (3.57039,5.16993) (3.26433,5.16993) (2.97876,5.16993) (2.71231,5.16993) (2.46371,5.16993) (2.23175,5.16993) (2.01533,5.16993) (1.81341,5.16993) (1.625,5.16993) (1.44921,5.16993) (1.2852,5.16993) (1.13216,5.16852) (0.989378,5.08068) (0.856155,4.99194) (0.731854,4.90236) (0.615877,4.81202) (0.507667,4.72097) (0.406703,4.62927) (0.3125,4.53698) (0.224606,4.44413) (0.142598,4.35077)
};

\plot coordinates{
(9.5,4.91886) (8.79685,4.91886) (8.14078,4.91886) (7.52865,4.91886) (6.95751,4.91886) (6.42462,4.91886) (5.92742,4.91886) (5.46351,4.91886) (5.03067,4.91886) (4.62681,4.91886) (4.25,4.91588) (3.89842,4.9015) (3.57039,4.87705) (3.26433,4.84414) (2.97876,4.80404) (2.71231,4.75778) (2.46371,4.70617) (2.23175,4.64991) (2.01533,4.58955) (1.81341,4.52558) (1.625,4.45841) (1.44921,4.38837) (1.2852,4.31579) (1.13216,4.2409) (0.989378,4.16394) (0.856155,4.08512) (0.731854,4.0046) (0.615877,3.92254) (0.507667,3.83908) (0.406703,3.75434) (0.3125,3.66844) (0.224606,3.58147) (0.142598,3.49352)
};
\plot coordinates{
(9.5,3.61471) (8.79685,3.61471) (8.14078,3.61471) (7.52865,3.61471) (6.95751,3.61471) (6.42462,3.61471) (5.92742,3.61471) (5.46351,3.61471) (5.03067,3.61471) (4.62681,3.61471) (4.25,3.61471) (3.89842,3.61471) (3.57039,3.61471) (3.26433,3.61471) (2.97876,3.61471) (2.71231,3.61471) (2.46371,3.61455) (2.23175,3.60621) (2.01533,3.58678) (1.81341,3.55809) (1.625,3.5216) (1.44921,3.47843) (1.2852,3.42952) (1.13216,3.3756) (0.989378,3.31731) (0.856155,3.25516) (0.731854,3.18961) (0.615877,3.12103) (0.507667,3.04974) (0.406703,2.97602) (0.3125,2.90011) (0.224606,2.82224) (0.142598,2.74258)
};
\legend{{$P_2=9, P_3 = 10$}, {$P_2=5, P_3 = 6$},  {$P_2=3, P_3 = 4$}}
\end{axis}
\end{tikzpicture}
\caption{Throughput with interfering incident links~\label{fig:eg:mac} with $P_s=20$}
\end{figure}
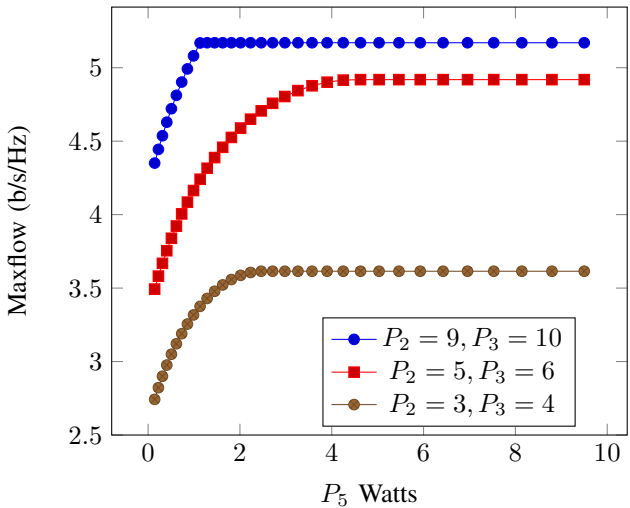

\section{Conclusions} \label{sec:conc}
In this paper, for the first time we propose an online algorithm for an arbitrary communication network that is representable by a directed acyclic graph and where all nodes are powered by EH. We show that that the proposed algorithm is optimal in terms of the competitive ratio, and the optimal competitive ratio is $2$. In the process of analysing the competitive ratio we consider a novel max-flow problem with logarithmic utilities and derive an optimal algorithm for it.